\documentclass[11pt]{article} % use larger type; default would be 10pt

\usepackage[utf8]{inputenc} % set input encoding (not needed with XeLaTeX)

\usepackage{geometry} % to change the page dimensions
\usepackage{graphicx} % support the \includegraphics command and options
\usepackage{amsfonts}
\usepackage{amsmath}
\usepackage{amsthm}
\newtheorem{theorem}{Theorem}

\newtheorem{lemma}[theorem]{Lemma}
\newtheorem{corollary}[theorem]{Corollary}

\theoremstyle{definition}
\newtheorem{definition}[theorem]{Definition}

\theoremstyle{remark}

\usepackage{booktabs} % for much better looking tables
\usepackage{array} % for better arrays (eg matrices) in maths
\usepackage{paralist} % very flexible & customisable lists (eg. enumerate/itemize, etc.)
\usepackage{verbatim} % adds environment for commenting out blocks of text & for better verbatim
\usepackage{subfig} % make it possible to include more than one captioned figure/table in a single float
\usepackage{fancyhdr} % This should be set AFTER setting up the page geometry
\usepackage{sectsty}
\allsectionsfont{\sffamily\mdseries\upshape} % (See the fntguide.pdf for font help)
\usepackage[nottoc,notlof,notlot]{tocbibind} % Put the bibliography in the ToC
\usepackage[titles,subfigure]{tocloft} % Alter the style of the Table of Contents

\newcommand{\comments}[1]{}
\usepackage[usenames]{color}

\title{Unbalanced allocations} 

\author{Amanda Redlich\thanks{Department of Mathematics, Bowdoin College, 8600 College Station, Brunswick ME 04011 ({\tt aredlich@bowdoin.edu}).  This material is based in part upon work supported by the National Science Foundation under Award No. 1004382.}}

\begin{document}
\maketitle

\begin{abstract}
We consider the \emph{unbalanced} allocation of $m$ balls into $n$ bins by a randomized algorithm using the ``power of two choices".  For each ball, we select a set of bins at random, then place the ball in the fullest bin within the set.  Applications of this generic algorithm range from cost minimization to condensed matter physics.  In this paper, we analyze the distribution of the bin loads produced by this algorithm, considering, for example, largest and smallest loads, loads of subsets of the bins, and the likelihood of bins having equal loads.
\end{abstract}

\section{Introduction}

Balanced allocations are a well-studied area in computer science.  A simple example is shoppers selecting cashiers at a grocery store; in a balanced allocation, most lines would be of the same (short) length.  More technical applications include allocating servers in a network, allocating disks for storage, and hashing (see e.g. \cite{azar}).  In all of these settings, the goal is to balance loads across all of the possible sites.

There are many situations in which the opposite is the case.  A simple example is selecting films at a multiplex;  most filmgoers are wary of short lines and choose the films with longer lines, assuming the popular ones are better.  This forms an \emph{unbalanced} allocation, with many empty loads and a few very large ones.  

There are many settings in which an unbalanced allocation appears.  For example, if the options are priced with a buy-at-bulk discount, the cost of an unbalanced allocation is much lower than a balanced one.  An unbalanced allocation also arises in natural processes; for example, condensed matter physics, smog and cloud formation, and galactic clustering in astrophysics (see e.g. \cite{pego}).

This paper presents a generic unbalanced allocation algorithm and analyzes its behavior.  This fundamental analysis provides a framework for developing algorithms that minimize cost or model natural behavior in specific settings.  Techniques used include differential equations, random walks, coupling, and witness trees.

The paper is structured as follows. The second section, following the introduction, defines an unbalanced allocation algorithm and gives some background on the subject of balanced and unbalanced allocations. Individual loads are analyzed next in the third section. In the fourth section, loads of subsets are defined and bounded. Motivated by these bounds, the fifth section discusses relationships between loads. A summary of present knowledge and an outline of future work comprise the last section.

\section{Definitions and background}

We first fix some convenient notation.  Throughout this paper, $m$ balls will be distributed into $n$ bins $B_1, \ldots, B_n$ over the course of $m$ time steps $t=1, 2, \ldots, m$.  We denote the load in $B_k$ at time $t$ as $b_k(t)$.  

The simplest allocation is to distribute the balls uniformly at random.
\begin{definition}
$\mathrm{UNIFORM}(m,n)$ places $m$ balls into $n$ bins by, at each time $t$, selecting $i_t$ uniformly at random from $[n]$ and placing the ball into $B_{i_t}$.
\end{definition}

The expected load of each bin under UNIFORM is $m/n$.  However, the expected maximum load is much larger; for example, under UNIFORM$(n,n)$, the expected maximum load is $\Theta(\log n / \log \log n)$.  The inspiration for this paper was Azar, Broder, Karlin and Upfal's balanced allocation algorithm, introduced in \cite{azar}.  Their insight was to use the ``power of two choices" to modify UNIFORM.  They randomly select several options at each time step, then place the ball in the least-loaded option.

\begin{definition}
$\mathrm{FAIR}(m,n,d)$ places $m$ balls into $n$ bins as follows.  At each time step, select $d$ choices from $[n]$ (uniformly randomly with replacement) to form a multiset $S_t$. Place the $t^{th}$ ball into $B_{\mu}$, where $b_{\mu}(t)=\min_{i \in S_t}\{b_i(t)\}$.  In the case of a tie, choose $B_{\mu}$ among the minimal bins uniformly at random.
\end{definition}

This significantly decreases the size of the largest load; for example, under \\FAIR$(n,n,2)$, the expected maximum load is $O(\log \log n)$.  For a full discussion of how well FAIR balances the loads, see \cite{azar} or \cite{mitz}.  Many variations on the original FAIR$(n,n,2)$ have been studied.  For example, $m$ much greater than $n$ (see e.g. \cite{heavy}), asymmetrical tie-breaking (see e.g. \cite{asymmetry}), or a non-uniform distribution on option sets (see e.g. \cite{graphs} and \cite{hypergraphs}).

These variations all generate balanced allocations.  Here, we study \emph{unbalanced} allocations using the power of two choices.  That is, we randomly select several options at each time step, then place the ball in the \emph{most}-loaded option.

\begin{definition}
$\mathrm{GREEDY}(m,n,d)$ places $m$ balls into $n$ bins as follows.  At each time step, select from $[n]$ uniformly randomly with replacement $d$ times to form a multiset $S_t$.  Place the $t^{th}$ ball into $B_{M}$, where $b_{M}(t)=\max_{i \in S_{t}} \{b_{i}(t) \}$.  In the case of a tie, choose $B_M$ among the maximal bins uniformly at random.
\end{definition}

In a dynamic queuing-theory setting, a similar algorithm to GREEDY is briefly discussed in \cite{russians}.  In the context of graph theory, it is a close relative to both preferential attachment models (see e.g. \cite{pref}) and Achlioptas processes (see e.g. \cite{giant}).  This algorithm also resembles the rich-get-richer process on Polya urns in \cite{rich}.  The key difference between GREEDY and the process in \cite{rich} is that probability of gaining a ball is related to the bin's load.  In GREEDY, the relationship between two bins' loads (rather than their actual values) is what determines probability.  

\section{Loads of individual bins}

The first research on FAIR studied the number of bins of small load sizes and the largest expected load.  We begin our analysis of GREEDY in the same way. In this section, we give explicit formulas for the expected number of bins of fixed load sizes, together with concentration bounds.  This improves on a technique of Mitzenmacher \cite{mitz} by using a theorem of Wormald \cite{worm}.  The theorem holds for $m=O(n^{4/3}$ and load sizes up to $k=O(n^{1/3}$.  We then turn to the load of the largest bin under GREEDY.  Surprisingly, this is no larger than that under UNIFORM.

\begin{theorem}\label{diffeq}
For the case $m=O(n^{4/3})$, $k=O(n^{1/3})$ and $d$ constant, the expected number of bins with load $k$ under GREEDY$(m,n,d)$ is $z_{k}(c)$ where $z_{k}$ satisfies the system of differential equations $$\{ z_{i}'(t)=2(z_{i-1}(t)+\cdots+z_{0}(t))^{d}-(z_{i-2}(t)+\cdots+z_{0}(t))^{d}-(z_{i}(t)+\cdots+z_{0}(t))^{d} \}_{i=2}^{k},$$ $$z_{1}'(t)=2z_{0}(t)^{d}-(z_{0}(t)+z_{1}(t))^{d},$$ $$ z_{0}'(t)=-z_{0}(t)^{d}$$ with initial values $z_{0}(0)=1$ and $z_{i}(0)=0$ for all $i>0$.  Furthermore, the number of bins with load $k$ is within $O(\lambda n)$ of its expectation with probability $1-O(\frac{k}{\lambda}e^{-n\lambda^3})$.
\end{theorem}

\begin{proof}
The proof of this theorem is a straightforward application of Theorem 5.1 of \cite{worm}, which gives tight estimates for sequences of random processes.  Here we consider the sequences of bins chosen as random processes, and define a function $y_{\ell}$ which, given a particular history of bin choices, outputs the number of bins of size $\ell$.  In other words, this theorem gives an estimate of $y_k$.    In the language of \cite{worm}, $S^{(n)+}$ is all possible sequences (of all lengths) of bins chosen.  Note that the range of the index $k$, i.e. of possible bin sizes, is a function of $m$ the number of balls distributed. Theorem 5.1 of \cite{worm} gives a relation between $y_k: S^{(n)+} \to \mathbb{R}$ and $f_k: \mathbb{R}^{a+2} \to \mathbb{R}$, where $a$ is some upper bound on $\ell$.  Here our $f_k(t/n, z_0, z_1, \ldots, z_a)$ will be an ``expected change polynomial",  $$n^{-d}\left( (z_0+\ldots+z_{\ell-1})^d-(z_0+\ldots+z_{\ell-2})^d-(z_0+\ldots+z_{\ell})^d+(z_0+\ldots+z_{\ell-1})^d\right).$$  

This function comes from calculating the expected change in the number of bins of size $\ell$ in one time step:  The number of bins of size $\ell$ increases by one when a bin containing $\ell-1$ balls is chosen, which happens with probability \\$n^{-d}\left( (z_0+\ldots+z_{\ell-1})^d-(z_0+\ldots+z_{\ell-2})^d\right)$.  On the other hand, the number of bins of size $\ell$ decreases by one when a bin containing $\ell$ balls is chosen, which happens with probability $n^{-d}\left( (z_0+\ldots+z_{\ell})^d-(z_0+\ldots+z_{\ell-1})^d\right)$.

We observe that the conditions of the theorem are satisfied for $C_0=1.1$, \\$D=\{(t, x_0, \ldots, x_a) \in [0,n^{1/3}]\times [0,1]^{a+1} \vert x_0+\cdots+x_a\leq 1, 0\leq t \leq n^{1/3}\}$, $\beta=1$, $\gamma=0$, $\lambda_1=0$, $a \leq n^{1/3}$, and Lipschitz constant $L=2d\binom{d}{d/2}$.  Briefly, this means that the change in $y_k$ at each time step is bounded, given that the past history of ball distributions $y_k$ behaves in a predictable manner, and the functions $f_k$ generate soluble differential equations on the domain $D$.  
Therefore, we see that the system of differential equations described in the theorem has a solution, and that solution approximates the number of bins with loads of size $k$.  More formally, 
$$y_{k}(t)=nz_{k}(t/n)+O(\lambda n)$$ with probability $$1-O(an\gamma+\frac{a\beta}{\lambda}e^{-n\lambda^3 / \beta^3})=1-O(\frac{a}{\lambda}e^{-n\lambda^3})$$ for any $\lambda>\lambda_1+C_0n\gamma=0$ and any $0\leq t \sigma n$, where $\sigma$ is the supremum of those $x$ to which the system's solution can be extended before reaching within $\ell^{\infty}$-distance $C\lambda$ of the boundary of $D$.     Note that the distance between internal point $(x, x_0, \ldots, x_{a})$ and the boundary of $D$ is $\max \{ (n^{1/3}-x), (1-\sum_{i=0}^{a} x_{i})/(a+1)\}$.  In particular, then, we have that $0\leq \sum_{i=0}^{a} x_{i} <1-C\lambda a$, and so $C\lambda a<1$.  Combining this condition with the bounds on probability gives meaningful bounds when $a < n^{1/3}$.
\end{proof}

Using this theorem, we can compare statistics for $\mathrm{FAIR}$, $\mathrm{UNIFORM}$, and $\mathrm{GREEDY}$ explicitly.  For example, the expected number of empty bins after $n$ balls have been distributed under these different allocations are given in Table 3.1, below (statistics for $\mathrm{FAIR}$ from \cite{mitz}).
\begin{table}[htdp]\footnotesize
\caption{Expected number of empty bins}
\begin{center}
\begin{tabular}{|c|c|c|c|}
\hline
                &$\mathrm{GREEDY}$         &$\mathrm{UNIFORM}$       &$\mathrm{FAIR}$\\
\hline
$d=2$     &$n/2$                                    &$n/e$                                    &$0.2384n$\\
\hline
$d=3$     &$n/\sqrt{3}$                          &$n/e$                                   &$0.1770n$\\
\hline
$d$          &$nd^{-1/(d-1)}$                &$n/e$                                             &-\\
\hline
\end{tabular}
\end{center}
\label{default}
\end{table}

Observe several important facts.  First, as expected, we have increased the number of empty bins.  Even for $d=2$, $\mathrm{GREEDY}$ has more empty bins than $\mathrm{UNIFORM}$ (and of course $\mathrm{FAIR}$).  Furthermore, this number increases as $d$ increases: as $d$ goes to infinity, this proportion approaches one.  Also note that we have an exact formula for the expected number of empty bins.  The analogous differential equation for $\mathrm{FAIR}$ can only be estimated computationally.

We now turn to bounding the largest load.  We cannot apply Theorem \ref{diffeq} directly, as it is not clear \emph{a priori} that the largest load is $O(n^{1/3})$.  Instead, we analyze GREEDY by coupling it with UNIFORM.  We show that GREEDY does not increase the load of the largest bin by more than a linear factor.  This is in contrast with FAIR, in which the largest bin is decreased exponentially from that of UNIFORM.
\begin{theorem}\label{dumb}
With probability greater than or equal to $1-1/n$, the largest load under $\mathrm{GREEDY}(cn,n,d)$ ($c$, $d$ arbitrary constants) is less than $\frac{(2+\epsilon)\log n}{\log \log n - \log d - \log c}$ for any constant $\epsilon>0$.
\end{theorem}

\begin{proof}
The key idea is to couple $\mathrm{UNIFORM}(dcn, n)$ with $\mathrm{GREEDY}(cn,n,d)$ by, under $\mathrm{UNIFORM}$, placing a ball in each of the bins in the option set under $\mathrm{GREEDY}$.  Thus the largest load under $\mathrm{UNIFORM}(dcn, n)$ is an upper bound on the largest load under $\mathrm{GREEDY}(cn,n,d)$.  

More formally, the coupling is that $(d(t-1)+i)^{th}$ ball is placed by $\mathrm{UNIFORM}(dcn,n)$ in the $i^{th}$ bin chosen for $S_{t}$ by $\mathrm{GREEDY}(m,n,d)$.  (If a bin is chosen twice  for $S_{t}$, then it receives two balls.)  This is clearly a valid coupling.  Furthermore, if bin $B_{i}$ receives a ball under $\mathrm{GREEDY}(cn,n)$ at time $t$, then $i \in S_{t}$, so bin $B_{i}$ also receives a ball under $\mathrm{UNIFORM}(dcn,n)$ at some time between $d(t-1)$ and $dt$.  Therefore once all $cn$ balls have been distributed under $\mathrm{GREEDY}(cn,n)$ and all $dcn$ balls have  been distributed under $\mathrm{UNIFORM}(dcn,n)$, every bin under $\mathrm{UNIFORM}$ is at least as full as its counterpart under $\mathrm{GREEDY}$.  So to bound the fullest bin in $\mathrm{GREEDY}(cn,n,d)$, it's enough to bound the fullest bin in $\mathrm{UNIFORM}(dcn,n)$.

All that remains is to bound $\mathrm{UNIFORM}$.  Consider $X'=\sum_{i=1}^{m} X_{i}'$, where the $X_{i}'$ are independent random
variables, each equal to 1 with probability $1-(1-1/n)^{d}$ and 0 with
probability $(1-1/n)^{d}$ (so $X'$ has the same distribution as a
count of the times some fixed bin $B$ was one of the $d$ options). Let
$X=\sum_{i=1}^{m} X_{i}$ be the sum of random variables, each independently equal
to 1 with probability $d/n$.  Note that $Pr(X \geq k) \geq Pr(X'
\geq k)$ for any $m$.  Furthermore, let $Y_{i}=Y_{i}-d/n$. Note that
$E(Y_{i})=0$ and $Y=\sum_{i=1}^{m}Y_{i}=X-m(d/n)=X-cd$.

We now use a standard bound on $Y$ (Theorem A.1.12 of \cite{joel}), which will lead to a bound on $X$, which in turn gives a bound on loads under $\mathrm{GREEDY}$.

\begin{lemma}\label{beta}
For $Y$ as defined above, and for arbitrary $\beta$, $$Pr( \left| Y \right| \geq (\beta-1)cd) < (e^{\beta-1}\beta^{-\beta})^{cd}.$$
\end{lemma}

Using this lemma, we will upper bound the
probability a particular bin is very large by $1/n^{2}$ (thus bounding the probability of any bin being very large by $1/n$).  Setting $\beta=\frac{(2 + \epsilon) \log n}{cd \log \log n}$ for arbitrary $\epsilon > \frac{2 \log
\log \log n}{\log \log n - \log \log \log 3}$ (for example, any constant
$\epsilon$ will do), it is a straightforward calculation to see that $Pr( \left| Y \right| \geq (\beta-1)cd)< 1/n^{2}$ for $n$ sufficiently large. On the
other hand,
$$Pr(\left| Y \right| \geq k) \geq Pr (X \geq k+cd) \geq Pr ( X' \geq
k+cd).$$  So the probability of a bin existing which has more than $\beta
cd=\frac{(2+\epsilon)\log n}{\log \log n}$ balls is thus $< 1/n$ for $n$
sufficiently large, and we have an upper bound on the largest load.
\qquad\end{proof}

The preceding proof assumed that the number of balls is linear in the number of bins, and also that the number of options is constant.  However, the same argument works with looser constraints.

\begin{theorem}\label{arbd}
If $$c(n)d(n)=o(\log n),$$ then with probability at least $1-1/n$ the most full bin under $\mathrm{GREEDY}(c(n)n,n, d(n))$ has load less than $$\frac{(2+\epsilon) \log n}{\log (2+\epsilon) + \log \log n - \log c(n) -\log d(n)}$$ for all constant $\epsilon>0$.
\end{theorem}
\begin{proof}
Let $a=\frac{(2+\epsilon)\log n}{c(n)d(n)}$ and $\beta = \frac{a}{\log a}$. When $c(n)$ and $d(n)$ are as stipulated in the theorem, for $n$ sufficiently large, $\beta$ is greater than 1.  So we may apply the same lemma to bound the probability that $X \geq \beta c(n)d(n)$ by
$(e^{\beta-1}\beta^{-\beta})^{c(n)d(n)}$.  As before,
the fullest bin almost surely contains fewer than
$$\beta c(n)d(n)=\frac{(2+\epsilon) \log n}{\log (2+\epsilon) + \log \log n -
\log c(n)-\log d(n)}$$ balls, as claimed.
\qquad\end{proof}

\section{Subsets of bins}
We now know that GREEDY creates a distribution with few small bins but no very large bins.  This presents us with a conundrum: where do the balls go?

The main goal of this section is to answer this question in a general setting.  Rather than studying how many balls are in individual bins, we study how many balls are in subsets of the bins.  This allows for bounds that hold even in the case of $md=\Omega(n\log n)$ (unlike Theorem \ref{arbd}) or $m=\Omega(n^{4/3})$ (unlike Theorem \ref{diffeq}). Our first result is a upper bound on the number of balls in the smallest $x$ subset of the bins, which holds for arbitrary $m,n,d,x$. 

\begin{theorem}\label{alpha1}
Under $\mathrm{GREEDY}(m,n,d)$, the expected number of balls in the smallest $xn$ bins is less than or equal to $x^dm$ for all values of $x$, $m$, $n$, and $d$.  The probability of the last $xn$ bins containing at least $k$ balls is upper bounded by \\$1-\sum_{i=0}^{k-1} \binom{m}{i}(1-x)^{m-i}x^i$.
\end{theorem}
\begin{proof}
The key idea of this proof is to label the bins according to their loads, rather than their original indices.  We  reorder bins from largest to smallest after each time step.  At time $t$, let $i(t)$ be the original index of the $i^{th}$ largest bin.  Break ties in this labeling randomly; for example, if $b_{1}(t)=3$, $b_{2}(t)=2$, and $b_{3}(t)=3$, then $1(t)$ is equally likely to be 1 or 3.  

The implementation of $\mathrm{GREEDY}$ under this labeling is the same as under the original labeling.  An option set $(i_{1}, \ldots, i_{d})$ is chosen uniformly randomly from $[n]^{d}$ at time $t$.  The bin $B_{i_{\mu}(t)}$ that gets the ball is such that $i_{\mu}=\min_{i_{j} \in S_{t} } \{ i_{j}\}$ and thus $b_{i_{\mu}(t)}(t)=\max_{i_{j} \in S_{t}} \{ b_{i_{j}(t)}(t) \}$.  

The reason for this reordering is that we now know that $i_{\mu}=\min \{i_{1}, \ldots, i_{d} \}$.   
In the original labeling, it was equally likely that $b_{1}(t)>b_{2}(t)$ or $b_{2}(t)>b_{1}(t)$.  Now, it is always true that $b_{1(t)}(t) \geq b_{2(t)}(t)$.  

On the other hand, it is no longer the case that giving the $t^{th}$ ball to bin $i$ implies bin $i$ is larger at time $t+1$.  For example, suppose the loads at time $t$ are $$(b_{1(t)}(t), b_{2(t)}(t),\ldots b_{6(t)}(t)) =(2,2,2,1,1,0)$$ and $S_{t}=\{3,4,6\}$.  Then the ball goes into bin $B_{3(t)}$.  At time $t+1$, the configuration is $$(b_{1(t+1)}(t+1), b_{2(t+1)}(t+1), \ldots b_{6(t+1)}(t+1))=( 3,2,2,1,1,0).$$  Although bin $B_{3(t)}$ was given the ball, $b_{3(t)}(t)=b_{3(t+1)}(t+1)$.  The increase is for $b_{1(t)}$:  $b_{1(t+1)}(t+1)=b_{1(t)}(t)+1$.

However, putting the ball into bin $B_{i_{\mu}}$ does guarantee that the increase is in a bin of index at most $i_{\mu}$.  The increased bin may move to the left after reordering, but never to the right.  That is, $$\sum_{j=1}^{i_{\mu}} b_{j(t)}(t) +1 =\sum_{j=1}^{i_{\mu}} b_{j(t+1)}(t+1).$$

We use this observation to get bounds on ball placement.  The load of the least $xn$ bins increases at time $t$ only if $i_{\mu}$, and therefore all of $S_{t}$, is contained within the least $xn$ bins.  In other words, $$\sum_{j>(1-x)n}b_{j(t)}$$ can increase only if $$i_{\mu} > (1-x)n.$$

The option set is within the least $xn$ bins at each time step with probability $x^{d}$.  Therefore the expected number of times this has happened, once all $m$ balls have been distributed, is $x^{d}m$.    Because the option set being within the least $xn$ is a necessary condition for the number of balls in the least $xn$ bins to increase, this tells us that the expected number of balls in the least $xn$ bins is at most $x^{d}m$.

Similarly, the least $xn$ bins contain at least $k$ balls only if the option set has been within the least $xn$ bins at least $k$ times.  Therefore the probability of the least $xn$ bins containing at least $k$ balls is at most $$1-\sum_{i=0}^{k-1}\binom{m}{i}Pr(S_{t}\text{ is in least }xn)^{i}Pr(S_{t}\text{ is not in least }xn)^{m-i}=1-\sum_{i=0}^{k-1}\binom{m}{i}x^{di}(1-x^{d})^{m-i}.$$
\end{proof}

This gives us another way to compare $\mathrm{GREEDY}$ and $\mathrm{UNIFORM}$.  As the expected proportion of bins of load $k$ under $\mathrm{UNIFORM}(m,n)$ is known to be $$\frac{(m/n)^{k}e^{-m/n}}{k!},$$ we can compute the expected fraction of smallest bins that contain a particular fraction of balls.  For example, when $m=n$, the expected number of bins of load 0 is $n/e$, of load 1 is $n/e$, and of load 2 is $n/2e$.  So if we take just the emptiest bins until we have half the balls, the expected number of bins would be $$n/e+n/e+(n/2-n/e)(1/2) \approx 0.8n.$$  We can use the same type of calculation to generate the following table.

\begin{table}[htdp]\footnotesize
\caption{Expected $x$ fraction of least bins containing $y$ fraction of balls, $m=n$}
\begin{center}
\begin{tabular}{|c|c|c|c|c|c|}
\hline
                &$\mathrm{UNIFORM}$         &$d=2$       &$d=3$	&$d=4$	&$d$\\
\hline
$y=1/3$     &$x=0.7$                                    &$x\geq 0.57$          &$x\geq 0.69$     &$x \geq 0.75$&$x\geq (1/3)^{1/d}$\\
\hline
$y=1/2$     &$x=0.8$                          &$x\geq 0.70$        &$x\geq 0.79$               &$ x \geq 0.84$            &$x\geq(1/2)^{1/d}$\\
\hline
$y=2/3$          &$x=0.88$                &$x\geq 0.81$		&$x\geq 0.87$                   &$x \geq 0.90$        &$x\geq(2/3)^{1/d}$\\
\hline
\end{tabular}
\end{center}
\label{default}
\end{table}

This answers our question from the beginning of this section:  $\mathrm{GREEDY}$ concentrates the balls in the largest few bins.  For the values in the table, $\mathrm{GREEDY}$ overtakes $\mathrm{UNIFORM}$ at $d=4$, but this effect becomes more pronounced as $d$ grows.  For example, when $d= 10$, at least half the balls are expected to be in the largest $(1-2^{-0.1})n\simeq 0.07 n$ bins. 

Our estimate for the number of balls in the last $x$ bins from the proof of Theorem \ref{alpha1} isn't necessarily tight.  It is possible that the option set is contained within the least $xn$ bins and yet their load does not increase.

For example, let $x=2/3$.  Recall the example in the proof of Theorem \ref{alpha1}: if the loads are $$(b_{1(t)}(t), b_{2(t)}(t),\ldots b_{6(t)}(t)) =(2,2,2,1,1,0)$$ and $S_{t}=\{ 3,4,6\}$, the ball goes into bin $B_{3(t)}$ and at time $t+1$, the configuration is $$(b_{1(t+1)}(t+1), b_{2(t+1)}(t+1), \ldots b_{6(t+1)}(t+1))=( 3,2,2,1,1,0).$$  Even though the option set was within the least $(2/3)n$ bins, the load of the least $(2/3)n$ bins didn't change.

We now find a lower bound on the number of balls in the least $xn$ bins which holds for arbitrary $m,n,d,x$.  The proof of this theorem uses a more in-depth analysis of possible configurations and is fairly involved.  For ease of exposition, we give a simplified proof here.

\begin{theorem}\label{alpha2}
For all values of $m$, $n$, and $d$, we have the following lower bounds on the expected number of balls under $\mathrm{GREEDY}(m,n,d)$ in the smallest $xn$ bins.  

When $xn>d$, the expected number is at least
$$ \left(\frac{1}{2ed}\right) x^{d+1}m.$$  

When $1< xn \leq d$, the expected number is at least $$e^{-d}x^{d}m.$$  

When $xn=1$, the expected number is at least $$x^{d+1}m.$$
\end{theorem}
\begin{proof}
The contradictory behavior in our example above happened because the first, second, and third bins had the same load.  When the $(xn+j)^{th}$ bin is given a ball, the $(xn-k)^{th}$ bin increases if and only if the bins between the $(xn-k)^{th}$ and the $(xn+j)^{th}$ all had the same load, and the $(xn-k-1)^{st}$ bin was larger.  More formally, when $$i_{m}(t)=xn+j,$$ then $$b_{(xn+j)(t)}(t)=b_{(xn+j)(t+1)}(t+1)$$ and $$b_{(xn-k)(t)}(t)+1=b_{(xn-k)(t+1)}(t+1),$$ if and only if $$b_{xn-k-1}(t)>b_{(xn-k)(t)}(t)=b_{(xn-k+1)(t)}(t)=\ldots=b_{(xn+j)(t)}(t).$$   In order to find a lower bound on the load of the last $x$ bins, we will bound the probability of a string of equally-loaded bins.

We here give the main ideas of the proof, using simple parameters.  The bounds in Theorem \ref{alpha2} are found by optimizing these parameters. We bound how many balls are in the least $xn$ bins by considering time steps when the option set is within the smallest $xn/3$ bins.  (As $n$ is tending to infinity, we ignore divisibility issues.) Call this type of time step ``good".  There are two possibilities for each good round: the increased bin will be within the least $2xn/3$ bins, or it won't.  Let $g$ be the number of good time steps.  Then the first or the second case will happen at least $g/2$ times.  If the first case happens at least $g/2$ times, then the last $xn$ bins will have at least $g/2$ balls, and we have a lower bound.  So all that remains is to find a lower bound on the number of balls in the last $xn$ bins if the second case happens at least $g/2$ times.

Consider the $(1-2x/3)n-1^{st}$ largest bin (i.e., the smallest bin outside of the least $2xn/3$).  Every time a case-two step occurs, that bin must be in a string of equally-loaded bins that stretches from some label greater than $(1-x/3)n$ to some label less than $(1-2x/3)n$.  That string's length decreases every time a case-two step occurs.  Therefore, after at most $(1-2x/3)n$ case-two steps, the string no longer contains the $(1-2x/3)n-1^{st}$ bin.  This means its load must have increased.  If there are $g/2$ case-two steps, then the $(1-2x/3)n-1^{st}$ bin must contain at least $$\frac{\text{number of case-two steps}}{(1-2x/3)n}=\frac{g}{2(1-2x/3)n}$$ balls.  There are $xn/3$ bins at least as full as that one within the last $xn$, so the least $xn$ bins must contain at least $$\left(\frac{xn}{3}\right)\left(\frac{g}{2(1-2x/3)n}\right)= \frac{gx}{2(3-2x)}$$ balls when the second case happens at least $g/2$ times.  

Combining the two possibilities, we see the least $xn$ bins contain at least $$(g/2) (\min \{1, x/(3-2x) \})=\frac{g}{6-4x}$$ balls overall.  The expected value of $g$ is $3^{-d}x^{d}m$, so the expected number of balls in the least $xn$ bins is at least $$3^{-d}x^{d+1}m/(6-4x).$$  

The full proof of Theorem \ref{alpha2} simply optimizes this argument by setting the three pieces of $x$, and the ratio of ``case one" to ``case two", to be unequal.  The ``optimal" breakdown of $xn$ into three pieces depends on $xn$, which is why three distinct bounds are given in the statement of Theorem \ref{alpha2}.  The breakdown is $(xn/d, 0, xn(1-1/d))$ in the first case, $(1, 0, nx-1)$ in the second case, and $(0,0,1)$ in the third.  The optimal ratio of cases is $$\frac{\alpha}{1-x+2\alpha}:\frac{1-x+\alpha}{1-x+2\alpha},$$ where $\alpha$ is the first term in the breakdown triples.
\qquad\end{proof}

Note that these two theorems hold regardless of $m$, $n$, and $d$.
For example, the expected number of balls in the least $\sqrt{n}$  bins under $\mathrm{GREEDY}$ is at most $n^{-d/2}m$.  We can also see that the smallest bin (i.e. the least $1/n$ fraction of bins) remains empty at least until $m =\Theta (n^{d})$ because $mx^{d}=m(1/n)^{d}=\frac{m}{n^{d}}$.

Both the upper and lower bounds hold for arbitrary $m$, $n$, and $d$.  However, they are not tight.  As we saw in the discussed examples, it is possible that an option set within the least $xn$ bins creates an increase in the greatest $(1-x)n$ bins.  This happens whenever there is a ``string" of equally-loaded bins that crosses the $xn$ boundary.  Furthermore, the specific location of the increase is determined by the length of the string.  For instance, return to the $(2,2,2,1,1,0)$ example at the beginning of this section; the increase happened in bin $B_{1(t+1)}$ because the equally-loaded string ended at bin $B_{1(t)}$. Understanding the behavior of such strings is therefore key to understanding the overall allocation.

\section{Relationships between bins}

The main goal of this section is to analyze the relative sizes of bin loads.  By definition, the behavior of GREEDY is  determined by the relative sizes of bin loads, not their absolute values.  As discussed above, understanding something as simple as when bins are equally loaded would be a big step towards understand the distribution.  We examine both relative loads and equal loads in this section.  First, we give a theorem about how bins' relative positions may change.

\begin{theorem}\label{gambler}
For any starting configuration of bins and balls, if bin $B_{i}$ has $ \delta n /(d-1)$ more balls than bin $B_{j}$, then the probability of bin $B_{i}$ becoming smaller than bin $B_{j}$ at any time in the future (i.e. after any number $m$ of balls has been added) under $\mathrm{GREEDY}$ is at most $e ^{-\delta}$.
\end{theorem}
\begin{proof}
The key idea of this proof is to view the changing gap between loads of the two bins as a random walk.  The walk reaching zero corresponds to the two bins having equal loads.  By using standard bounds on the probability of a random walk reaching zero, we are able to bound the probability of the bins becoming equally loaded, and thus the probability of the bins swapping relative positions.

Fix two bins, without loss of generality $B_{1}$ and $B_{2}$, and consider $\left| b_{1}-b_{2} \right|$ at each time step.  For most steps, this gap doesn't change; usually a ball is placed in neither $B_{1}$ nor $B_{2}$.  So we condition on one of those two bins getting a ball.  If the larger bin gets a ball, $\left| b_{1}-b_{2} \right|$ increases by 1.  If the smaller bin gets a ball, $\left| b_{1}-b_{2} \right|$ decreases by 1.  This is a random walk with a reflecting barrier at 0.  Since a larger bin is more likely to get a ball than a smaller bin, it is biased in favor of +1.

We now bound the bias.  Suppose $B_{1}$ and $B_{2}$ are currently ranked the $i^{th}$ and $j^{th}$ bins, with $i>j$, where the $1^{st}$ bin is the smallest and the
$n^{th}$ bin is the biggest.  (This is the opposite of our earlier convention, but makes the following computations much simpler.)  
The probability, given that one of
the two bins gets a ball, of the bigger bin getting it is
$$\frac{i^{d}-(i-1)^{d}}{i^{d}-(i-1)^{d}+j^{d}-(j-1)^{d}}.$$

We now minimize this probability.  First, note the minimum must be at $i$ and $j$ such that $j+1=i$; the closer together two bins are, the closer together their respective probabilities of getting a ball are.  So it is enough to minimize $$\frac{i^{d}-(i-1)^{d}}{i^{d}-(i-1)^{d}+(i-1)^{d}-(i-2)^{d}}=\frac{i^{d}-(i-1)^{d}}{i^{d}-(i-2)^{d}},$$ which happens at $i=n$.  That is, the probability of the larger of two bins getting a ball is minimized when they are the largest and second largest bins.  In that case the larger bin gets a ball with probability $$\frac{n^{d}-(n-1)^{d}}{n^{d}-(n-2)^{d}}.$$

We will now formalize the coupling of a random walk with the load gap.  For ease of notation, let $t_{k}$ be the $k^{th}$ time step at which bin $B_{1}$ or $B_{2}$ gets a ball.  Let $X_{k}$ be the position of a random walk with bias $\epsilon$ at time $k$, where $$\frac{1+\epsilon}{2}=\frac{n^{d}-(n-1)^{d}}{n^{d}-(n-2)^{d}}.$$  That is, the random walk has the same bias as that between the largest and second largest bins.

We now couple the sequence $\left| b_{1}(t_{1})-b_{2}(t_{1})\right|, \left| b_{1}(t_{2})-b_{2}(t_{2})\right|, \ldots$ with $X_{1}, X_{2}, \ldots$ so that the random walk takes a -1 step every time the gap shrinks and may also take a -1 step even if the gap increases, in such a way that the probability of a +1 step in the random walk is always $\frac{1+\epsilon}{2}$.  That is, if at time $t_k$ the larger bin has probability $\gamma$ of being chosen over the smaller bin, then if $$\left| b_{1}(t_{k})-b_{2}(t_{k})\right| > \left| b_{1}(t_{k-1})-b_{2}(t_{k-1})\right|,$$ $$X_{k-1}+1=X_{k}$$ with probability $$\frac{1+\epsilon}{2\gamma}$$ and $$X_{k-1}-1=X_{k}$$ with probability $$1-\frac{1+\epsilon}{2\gamma}.$$  If $$\left| b_{1}(t_{k})-b_{2}(t_{k})\right| < \left| b_{1}(t_{k-1})-b_{2}(t_{k-1})\right|,$$then $$X_{k}=X_{k-1}-1$$ with probability 1.  

If $B_{1}$ and $B_{2}$ switch relative positions, there exists a time step $t$ at which $b_{1}(t)=b_{2}(t)$.  Therefore they switch positions only if there exists $t$ such that $\left|b_{1}(t)-b_{2}(t)\right|=0$.  The coupling above shows that if the bin load gap reaches 0 at time $t$, the random walk must have also reached 0 at time $t$ or earlier.  We now bound the probability that the random walk reaches 0.

This bound uses a gambler's ruin argument (see e.g. \cite{feller} for more details).  The probability of ruin starting from position $x$ is $$\left( \frac{(1-\epsilon)/2}{(1+\epsilon)/2}\right)^{x}=\left(\frac{1-\epsilon}{1+\epsilon}\right)^{x}=\left(\frac{n^{d}-(n-1)^{d}}{(n-1)^{d}-(n-2)^{d}}\right)^{x}.$$    Therefore to bound this probability by $e^{-\delta}$, we need 
$$x\geq \frac{\delta}{\log\left(\frac{n^{d}-(n-1)^{d}}{(n-1)^{d}-(n-2)^{d}}\right)}.$$  Notice that as 
$n\to \infty$, 
$\frac{n^{d}-(n-1)^{d}}{(n-1)^{d}-(n-2)^{d}}\to1+\frac{d-1}{n-(3/2)(d-1)}$.  Therefore 
\\
$\log \left( \frac{n^{d}-(n-1)^{d}}{(n-1)^d-(n-2)^{d}}\right)\sim \frac{d-1}{n-(3/2)(d-1)}$, and $$
\frac{\delta}{\log\frac{n^{d}-(n-1)^{d}}{(n-1)^{d}-(n-2)^{d}}} \sim \frac{\delta (n-(3/2)(d-1))}{(d-1)} <\frac{\delta n}{(d-1)}.$$

So the gambler's ruin argument shows that for $x > \delta n/(d-1)$ and $n$ sufficiently large, a random walk starting at $x$ with bias $\epsilon$ will reach 0 with probability less than $e^{-\delta}$.  By coupling this walk with the load gap of the bins, we see the probability of two bins switching position under $\mathrm{GREEDY}(m,n,d)$ is less than $e^{-\delta}$ if they start with loads at least $\delta n/(d-1)$ apart.
\qquad\end{proof}

In other words, Theorem \ref{gambler} tells us that bins' relative orders stabilize once the gaps between them are linear in $n$.  We can combine this with our previous results in specific cases to find when relative positions should stabilize.

For example, if  $m=\delta n^{2}/(d-1)$ it is unlikely that the least-loaded bins will overtake the heaviest-loaded:  For this value of $m$, the largest bin must have load at least $\delta n/(d-1)$.  On the other hand, the probability of the total load in the least $xn$ bins being greater than $x^{d}mn$ is at most $1/n$ by Markov's inequality and Theorem \ref{alpha1}.  So with probability $1-1/n$ these bins have at most $x^{d}mn$ balls overall.  When $d >4$, we may let $x=\left( \frac{d-1}{\delta}\right)^{2/d}n^{-4/d}$ and thus $x^{d}mn<1$.  Then the least $x n = \left( \frac{d-1}{\delta}\right)^{2/d}n^{1-4/d}$ bins are empty with probability at least $1-1/n$.  Therefore we can apply Theorem \ref{gambler} to see that once $\delta n^{2} / (d-1)$ balls have been allocated, each of the least $xn$ will become the largest bin at any time in the future with probability at most $1-(1-e^{-\delta})(1-1/n) \sim e^{-\delta}$.

We now know that, although GREEDY is defined in terms of relative bin loads, in fact the absolute differences in bin loads drive GREEDY's behavior.  A difference of 0 corresponds to equally-loaded bins, which are key in the proofs of Theorems \ref{alpha1} and \ref{alpha2}, and a significant difference in bin loads is exactly the condition necessary to apply Theorem \ref{gambler}.  This motivates our final set of results.  We give conditions under which the number of equally-loaded bins is bounded from above, and extend this to conditions under which most gaps between loads are bounded from below.   These theorems are a significant step towards understanding GREEDY in full generality. The proofs of both of these theorems use a lemma about how choices at an early time step can affect the final allocation.

 \begin{lemma}\label{paradox}
 For any allocation of balls $\mathbf{b}=(b_{1}, b_{2}, \ldots b_{n})$, and any $i \neq j$, consider the results of $\mathrm{GREEDY}(m,n,d)$ on initial configurations of $\mathbf{b}+\mathbf{e}_{i}$ and $\mathbf{b}+\mathbf{e}_{j}$ (here $\mathbf{e}_{i}$ and $\mathbf{e}_{j}$ are the standard unit vectors with 1 in the $i^{th}$ ($j^{th})$ position and 0 elsewhere).  For $m=O( n \log n)$, the final load of $B_{i}$ starting from $\mathbf{b}+\mathbf{e}_{i}$ will be greater than the final load of $B_{i}$ starting from $\mathbf{b}+\mathbf{e}_{j}$ with high probability.  In particular, for $m=cn\log n$, the final load from $\mathbf{b}+\mathbf{e}_{i}$ will be greater with probability at least  $1-n^{cd^{2}-1}$.
 \end{lemma}

This lemma might seem obvious; of course placing a ball in bin $B_i$ at time $t$ should increase the load of bin $B_i$ at time $t+cn\log n$.  However, this is not always the case.  For example, suppose $\mathbf{b}=(1,1,1)$, $i=2$, and $j=1$ (so we are comparing $\mathrm{GREEDY}$ on $\mathbf{b}+\mathbf{e}_{1}=(2,1,1)$ and $\mathbf{b}+\mathbf{e}_{2}=(1,2,1)$).  Further suppose that the option sets are $S_{1}=S_{2}=\{ 1, 3\}$ and $S_{3}=S_{4}=\{2,3\}$.  The following table shows possible outcomes.

\begin{table}[htdp]
\caption{Loads under different configurations and tie breaks}
\begin{center}
\begin{tabular}{|c|c||c|c||c|c|}
\hline
 &&$B_{2}>B_{3}$&$B_{3}>B_{2}$& $B_{1}>B_{3}$&$B_{3}>B_{1}$\\
 \hline
$t=0$&-&$211$&$211$&$121$&$121$ \\
\hline
$t=1$&$S_1=\{1,3\}$&$311$&$311$&$221$&$122$\\
\hline
$t=2$&$S_{2}=\{ 1, 3\}$&$411$&$411$&$321$&$123$\\
\hline
$t=3$&$S_3=\{2,3\}$&$421$&$412$&$331$&$124$\\
\hline
$t=4$&$S_{4}=\{2,3\}$&$431$&$413$&$341$&$125$\\
\hline
\end{tabular}
\end{center}
\label{default}
\end{table}

  Recall that $\mathrm{GREEDY}$ breaks ties uniformly at random.  In the above table, $B_{i}>B_{j}$ indicates that the tie between bins $B_{i}$ and $B_{j}$ is broken in favor of bin $B_{i}$.  Note that a tie is broken at time $t=3$ for initial configuration $(2,1,1)$ and at time $t=1$ for initial configuration $(1,2,1)$.  We see that, if the tie is broken in favor of $B_{2}$ at time $t=3$ and in favor of $B_{3}$ at time $t=1$, $b_{2}(4)=3$ starting from $(2,1,1)$ and $b_{2}(4)=2$ starting from $(1,2,1)$.  The effect of bin $B_{2}$ being larger initially is to make $B_{2}$ smaller after more balls have been placed.  So in fact Lemma \ref{paradox} is nontrivial.

\begin{proof}
 Our example paradox relied on the option sets intersecting.  Ball placement at time 1 influenced placement at time 2, for example, because the option sets at times 1 and 2 were the same.  In general, an extra ball in $B_{i}$ can cause $B_{j}$ to increase only if there is an intersection or chain of intersections between the option sets containing $B_{i}$ and the option sets containing $B_{j}$.  For example, if the option sets are $\{3,i\}, \{3,4\}, \{5,7\}, \{4,j\}$, it is possible that the behavior of $B_{i}$ can influence $B_{j}$; the decision between $i$ and 3 made for $S_1$ affects the decision between 3 and 4 for $S_2$, which affects the decision between 4 and $j$ for $S_4$.  To bound the probability of a paradox, we study the structure of option set intersections.
 
   We call the elements affected by the choice of $B_i$ or $B_j$ an ``influence set".   In our previous example with option sets $\{3,i\}, \{3,4\}, \{5,7\}, \{4,j\}$, the influence set is $\{i, 3, 4, j\}$; although 5 and 7 appear as options, there is no intersection or chain of intersections for $B_i$ and $B_j$ that contain 5 or 7.

For ease of notation, assume we are comparing a placement in Bin 1 with Bin 2 (so the starting configurations are $\mathbf{b}+\mathbf{e}_{1}$ and $\mathbf{b}+\mathbf{e}_{2}$).  Let $T_{t}$ be the influence set at time $t$.  Initially, $T_{0}=\{1, 2 \}$.  Given a sequence of option sets $\{S_{t}\}$, we can define $T_{t}$, the influence set at time $t$, recursively.  $$T_{t}=T_{t-1} \cup \{ x \vert x \in S_{t} \text{ and } S_{t} \cap T_{t-1} \neq \emptyset \}$$ In our example, $T_0=\{i,j\}, T_1=\{i,j, 3\}, T_2=\{i,j,3,4\}, T_3=\{i,j,3,4\}, T_4=\{i,j,3,4\}$.
 
As observed earlier, a paradox may arise only if a subsequent option set $S_{r}$ contains $B_{1}$ or $B_{2}$ and some other bin which was already influenced by the initial choice of $\mathbf{e}_{1}$ or $\mathbf{e}_{2}$.  (In our example, this happens for $S_4$.)  We now bound the probability of this happening.
 
 If $S_{r}$ does contain both some index that is in $T_{r-1}$ and also $B_{1}$ or $B_{2}$, then there must be a subsequence of option sets $S_{t_{1}}, S_{t_{2}}, \ldots S_{t_{s-1}}, S_{t_{s}}=S_{r}$ such that for all $i \in [s]$, $S_{t_{i}}$ has a non-empty intersection with $S_{t_{i-1}}$, and $S_{t_{1}}$ contains $B_{1}$ or $B_{2}$.  There are $\binom{r}{s}$ choices for indices of a subsequence of length $s$.  The probability that any particular length-$s$ subsequence is intersecting is bounded by $(d^{2}/n)^{s}$.  There are two choices for $S_{t_{1}}$ and $S_{r}$ (to contain $B_{1}$ or $B_{2}$), and the probability of either is less than $d^{2}/n$.  So the overall probability of an intersecting subsequence of length $s$ is bounded by $4\binom{t}{s}(d^{2}/n)^{s+1}$.
 
 Then for the existence of such a sequence of any length, we have the bound $$4\sum_{s=1}^{t} \binom{t}{s}(d^{2}/n)^{s+1}\leq (4d^{2}/n)(1+d^{2}/n)^{t} \leq (4d^{2}/n)e^{td^{2}/n}.$$  Note that when $t= c n \log n$ for $c$ any constant, this is $O(n^{cd^{2}-1})$.  In particular, if $cd^{2}<1$, this is $o(1)$.
 \qquad\end{proof}

With this lemma in hand we are ready to prove Theorem \ref{equal}.  
\begin{theorem}\label{equal}
For any $\delta, m,n,d$, if $\alpha, \beta, \epsilon,\lambda,m',t$ are such that 
 \begin{itemize}
\item $\frac{(\delta^{d}m'-\alpha)\delta}{(6-4\delta)n}-\frac{2^{d}\epsilon^{d}m'+\beta}{\epsilon n}>\epsilon dt+\lambda$
 \item $m'+t=m$
 \end{itemize}
 then under $\mathrm{GREEDY}(m,n,d)$, for any $\gamma$, any pair of bins outside a set of size $\delta n$ are not equal with probability at least 
$$\scriptstyle{(1-e^{-2\alpha^2/m'})(1-e^{-2\beta^{2}/m'})
(1-e^{-2\lambda^{2}/t})(1-2e^{-(\gamma^2 n)/(2^{d+3}\epsilon^{d-1}t)})
(1-(4d^{2}/n)e^{td^{2}/n})(1-\sqrt{2/\pi((2d\epsilon^{d-1}t/n)-\gamma)})}.$$
 \end{theorem}
 
 Loosely, this theorem states that when $m$ is large and $d$ is small, an arbitrary pair of bins is unlikely to be equally loaded.  The proof has several stages.  We first analyze the option sets and determine which types of option sets have an effect on the final loads, as in Lemma \ref{paradox}.  We then count the number of significant option sets, again as in Lemma \ref{paradox}.  Finally, we bound how likely it is that a sequence of option sets will have the wrong effect.

 \begin{proof} 
Let the number of balls to be distributed be $m= m' +t$.  We fix two bins, $A$ and $B$, and bound the probability that they have the same loads after all $m$ balls have been allocated.  We allocate the balls in two phases.  In the first phase, allocate $m'$ balls. The ``exceptional" set of size $\delta n$ is determined at this point.  We then analyze the effect of the remaining $t$ steps on $A$ and $B$, assuming they are in the set of $(1-\delta)n$ unexceptional bins. 
 
 In analyzing the last $t$ steps, we use $W_{\epsilon}$, a set of bins that are much smaller than bin $A$ or bin $B$ at time $m'$.  For ease of notation, let $a(m')$ (or $b(m')$ )be the loads of $A$ (or $B$) at time $m'$.  Let $W_{\epsilon}$ be the least $\epsilon n$ bins at time $m'$.  We choose $\epsilon$ such that bins in $W_{\epsilon}$ each have loads at most $\min \{a(m), b(m)\}-g$.  That is, the ``gap" between the loads of $A$ and $B$ and any bin in the least $\epsilon n$ is at least $g$.
 
We are now ready to consider the final $t$ rounds.  We first reveal the rounds with option sets that either don't contain $A$ or $B$, or contain $A$ or $B$ and at least one bin not in $W_{\epsilon}$.  Call the remaining option sets ``important".  That is, $S$ is important if $S \subseteq W_{\epsilon} \cup \{A\}$ or $S \subseteq W_{\epsilon} \cup B$.  We choose gap size $g$ so that over the course of the final $t$ rounds, bins in $W_{\epsilon}$ are likely to remain below the loads of $A$ and $B$.  Therefore if an option set is important, it is likely that $A$ or $B$ gets a ball.  Let $q$ be the number of rounds with important option sets.

Now reveal all the non-$A$ or $B$ elements of the $q$ important option sets.  Each important option set contains exactly one of $A$ or $B$.  So there are $2^{q}$ possibilities, $\{A, B\}^{q}$, once the other elements are revealed.  Create a partial ordering $<_{AB}$ by setting $A < B$ (so, e.g., $ABBAB <_{AB} BBBAB$).  

With probability $1-O((d^2/n)e^{td^2/n})$, this partial ordering corresponds to $<_{\ell}$, ordering by bin loads, where  $\mathbf{v}<_{\ell}\mathbf{u}$ if the $\mathbf{v}$ sequence of $A$s and $B$s would result in fewer balls in bin $B$ and more in bin $A$ than the $\mathbf{u}$ sequence. To see this, recall Lemma \ref{paradox}.  The probability of a chain of intersection within the option sets is at most $(4d^{2}/n)e^{td^{2}/n}$.  Given that there is no chain of intersection, the placement of a ball into $A$ or $B$ at any of the important steps does not increase the load of the other bin.  If $\mathbf{v}<_{AB}\mathbf{u}$, the $\mathbf{v}$ sequence will generate a smaller $B$ than the $\mathbf{u}$ sequence.
 
Therefore the set of revealed sequences such that $A$ and $B$ have the same load is an anti-chain under $<_{AB}$ with probability $1-O((d^2/n)e^{td^2/n})$.  By Sperner's Lemma, it has size at most $\binom{q}{q/2}$ with the same probability.   

Now put these assumptions together.  If there are exactly $q$ important option sets, if the gap between the last $\epsilon n$ bins and $A$ and $B$ is as expected, and if the sequence of important sets is non-intersecting, the probability of $A$ and $B$ having the same number of balls after all $m$ balls have been distributed is at most $\binom{q}{q/2}/2^{q}$.   
 
It remains to find values for $\epsilon$, $\delta$, $g$, $t$, $q$, and $\gamma$.  First consider $\delta$.  Note that, if $A$ and $B$ are in the upper $1-\delta$ proportion of bins at time $m'$, then $A$ and $B$ have at least as many balls as the $\delta n^{th}$ bin at time $m'$.  That bin has at least as many balls as the average of the least $\delta n$ bins' loads.  

Recall the proof of Theorem \ref{alpha2} used the expected number of ``good" steps and multiplied it by a correction factor to discount the times a ball placed within the least $\delta n$ bins moved outside the least $\delta n$.  In fact if the number of good steps is $s$, then the number of balls in the least $\delta n$ is at least $s\delta/(6-4\delta)$.   Note that the option sets are distributed uniformly, so we can use the standard Chernoff bound $$Pr(s<\delta^{d}m'-\alpha)<e^{-2\alpha^{2}/m'}$$ to see that the number of balls in the last $\delta$ bins after the first $m'$ steps is near the expectation, i.e. at least $(\delta^{d}m'-\alpha)\delta/(6-4\delta)$, with probability at least $1-e^{-2\alpha/m'}$.  Because the $\delta n^{th}$ bin has at least as many balls as the average, if $A$ and $B$ are in the upper $(1-\delta)n$ bins, they will each have at least $$(\delta^{d}m'-\alpha)/(n)(6-4\delta)$$ balls with probability at least $$1-e^{-2\alpha^{2}/m'}.$$

 Now turn to $\epsilon$.  Suppose the number of balls in the least $2\epsilon n$ bins at time $m'$ is $x$.  Then the $\epsilon n^{th}$ smallest bin would contain at most $x/\epsilon n$ balls. 
Recall that $x$ is at most the number of times the option set is within the last $2\epsilon n$.  The Chernoff bound tells us that $$Pr (x > (2\epsilon)^{d}m'+\beta)< e^{-2\beta^{2}/m'}.$$  So the least $\epsilon n$ bins will each have loads at most $$\frac{2^{d}\epsilon^{d}m'+\beta}{\epsilon n}$$ with probability at least $$1-e^{-2\beta^{2}/m'}.$$
 
 We can combine these two results to see that, with probability at least $$(1-e^{-2\alpha^2/m'})(1-e^{-2\beta^{2}/m'}),$$ the gap between $A$ or $B$ and any bin within the least $\epsilon n$ will be at least $$\frac{(\delta^{d}m'-\alpha)\delta}{(6-4\delta)n}-\frac{2^{d}\epsilon^{d}m'+\beta}{\epsilon n}.$$ 
 
 Recall that we assumed the gaps between $A$ or $B$ and $W_{\epsilon}$ at time $m'$ were so large that $A$ and $B$ would still be larger than $W_{\epsilon}$ at time $m'+t$.  We now determine exactly how large a gap is necessary to guarantee it will not be closed after $t$ steps.  
 
 Each bin in $W_{\epsilon}$ may increase only if it is a member of an option set.  So it is enough to bound the number of times any bin in $W_{\epsilon}$ appears in an option set during the last $t$ time steps.  Again, note that the number of times $x$ that bins in $W_{\epsilon}$ are an option can be Chernoff bounded: $Pr(x>\epsilon dt+\lambda)\leq e^{-2\lambda^{2}/t}$.
 
 So overall, with probability at least $$(1-e^{-2\alpha^2/m'})(1-e^{-2\beta^{2}/m'})(1-e^{-2\lambda^{2}/t})$$ we know the gap will be bigger than the increase in the smallest bins: $$\frac{(\delta^{d}m'-\alpha)\delta}{(6-4\delta)n}-\frac{2^{d}\epsilon^{d}m'+\beta}{\epsilon n}> \epsilon dt+\lambda.$$ 
 
 We now turn to $t$, $q$, and $\gamma$. 
 The probability $p$ of an option set being important is $$2d\epsilon^{d-1}/n \leq p=2((\epsilon+1/n)^{d}-\epsilon^{d})\leq 2(2^{d}-1)\epsilon^{d-1}/n.$$  Thus applying Lemma \ref{beta} (Theorem A.1.12 in \cite{joel}) to the sum of $t$ random variables, each with success probability $2(2^{d}-1)\epsilon^{d-1}/n$, by setting 
 $$\beta=1+\frac{\gamma n}{2(2^d-1)\epsilon^{d-1}t},$$ is enough to tell us that $q$, the total number of important option sets, is more than $\gamma$ above the expectation with probability at most $$e^{-\gamma^{2}n/2(2^d-1)\epsilon^{d-1}t}.$$
 
 Similarly, we may consider the sum of $t$ random variables, each with success probability $2d\epsilon^{d-1}/n$, and apply the following lemma (Theorem A.1.13 in \cite{joel}).
 \begin{lemma}
 For $X$ as above, $Pr[X<-a]<e^{-a^{2}/2pt}$.
 \end{lemma}

 Therefore $q$ is less than $\gamma$ below the expectation with probability at most $$e^{-\gamma^{2}n/2d\epsilon^{d-1}}.$$  We may bound the overall probability that $q$ is more than $\gamma$ off from its expected value by $$e^{-\gamma^{2}n/2(2^d-1)\epsilon^{d-1}t}+e^{-\gamma^{2}n/2d\epsilon^{d-1}}\leq 2e^{-(\gamma^2 n)/(2^{d+3}\epsilon^{d-1}t)}.$$

 The probability that the important sets are non-intersecting is, as computed in the proof of Theorem \ref{paradox}, at least $1-(4d^{2}/n)e^{td^{2}/n}$.  Given that the important sets are non-intersecting, the probability of the important set sequence making $A$ and $B$ have equal loads is at most 
 $$ \frac{\binom{q}{q/2}}{2^{q}}\sim \frac{\sqrt {2}}{\sqrt{\pi}q}
 \leq \frac{\sqrt{2}}{\sqrt{\pi((2d\epsilon^{d-1}t/n)-\gamma)}}.$$
   So overall we have probability $$\scriptstyle{(1-e^{-2\alpha^2/m'})(1-e^{-2\beta^{2}/m'})
(1-e^{-2\lambda^{2}/t})(1-2e^{-(\gamma^2 n)/(2^{d+3}\epsilon^{d-1}t)})
(1-(4d^{2}/n)e^{td^{2}/n})(1-\sqrt{2/\pi((2d\epsilon^{d-1}t/n)-\gamma)})}$$ of the two bins having the same load, as desired.
 \qquad\end{proof}

In fact, we can generalize the above proof to show that arbitrary gaps of constant size are unlikely:
\begin{theorem}\label{gaps}
  For any $\delta, j, m,n$, and $f(x,d,n,m)$ the bound on expected value given in Theorem \ref{alpha2}, if  $\alpha, \beta, \epsilon,\lambda,m',t$ are such that 
 \begin{itemize}
\item $\frac{(\delta^{d}m'-\alpha)\delta}{(6-4\delta)n}-\frac{2^{d}\epsilon^{d}m'+\beta}{\epsilon n}>\epsilon dt+\lambda$
 \item $m'+t=m$
 \end{itemize}
 then under $\mathrm{GREEDY}(m,n,d)$, for any $\gamma$, any pair of bins outside a set of size $\delta n$ are at least $j$ balls apart from each other with probability at least $$\scriptstyle{(1-e^{-2\alpha^2/m'})(1-e^{-2\beta^{2}/m'})
(1-e^{-2\lambda^{2}/t})(1-2e^{-(\gamma^2 n)/(2^{d+3}\epsilon^{d-1}t)})
(1-(4d^{2}/n)e^{td^{2}/n})(1-j\sqrt{2/\pi((2d\epsilon^{d-1}t/n)-\gamma)}).}$$
 \end{theorem}

 \begin{proof}
 The same argument holds.  The first four terms guarantee that the important option sets are indeed important.  The fifth term, as before, guarantees that these sets do not intersect, so that they generate a partial ordering on $\{A,B\}^{q}$.  Now we want to find the vectors which generate $A$ and $B$ within $j$ of each other.  Note that for any difference $i$, the vectors that produce a difference $i$ between $A$ and $B$ form an antichain.  So the probability of any fixed difference is, as before, bounded by $\binom{q}{q/2}/2^{q}$.  Summing over $i$ from 0 to $j-1$ gives the sixth term in the product, and the theorem is proved.
\qquad \end{proof}
 
 By setting appropriate values for all variables, we can develop several corollaries of these theorems. We first give specific values that guarantee most bin pairs are not equally loaded.  We then use these values to give a bound on the possible number of bins of any one load.

  \begin{corollary}\label{numbers}
 For $m$ and $n$  such that $n^{2}\log n =o(m)$, for any constant $\delta$, any pair of bins outside of a set of size $\delta n$ are equal with probability $O(1/\sqrt{\log n})$.
 \end{corollary}
  \begin{proof}
We first verify that Theorem \ref{equal} may be applied. Let $t=(1/2d^{2})n\log n$ and $m'=m-t$.  Let $\alpha=\beta=\sqrt{m'\log n}$, $\lambda=\sqrt{t \log n}$.  Let $\gamma=(\log n)^{3/4}$.  Let $\epsilon=\delta^{(d+1)/(d-1)}/32$.  Then $$\frac{(\delta^{d}m'-\alpha)\delta}{(6-4\delta)n}-\frac{2^{d}\epsilon^{d}m'+\beta}{\epsilon n} \sim \left(\frac{\delta^{d+1}}{6-4\delta}-2^d\epsilon^{d-1}\right)\frac{m'}{n}.$$  Notice that $6-4\delta\leq 6$ and, since $d\geq 2$, $2^{d-5d+5}\leq 1/8$ to see that this is lower bounded by $$\delta^{d+1}(1/24)m'/n.$$  On the other hand, $$\epsilon d t +\lambda \sim (\delta^{(d+1)/(d-1)}/32)n \log n.$$  Because $n^{2} \log n = o(m)$, this is enough to show that the first condition of Theorem \ref{equal} is satisfied.  The second condition is satisfied by definition.  All that remains is a simple computation of the probability bound of Theorem \ref{equal} with the given values, which indeed shows an error probability of $O(1/\sqrt{\log n})$.
  \qquad \end{proof}
 
 \begin{corollary}
Outside of a set of size $\delta n$ for any constant $\delta$, with high probability, for $m$ such that $n^{2}\log n=o(m)$,  the greatest number of bins with the same load is less than $y$, for any $y$ such that $\frac{n}{(\log n)^{1/4}}=o(y)$.  For example, with high probability there are no more than $\frac{n}{(\log n)^{1/5}}$ bins with equal loads outside of a set of size $\delta n$.
 \end{corollary}
 \begin{proof}
 As in Theorems \ref{equal} and \ref{gaps} and Corollary \ref{numbers}, we exclude $S_{\delta}$, the $\delta n$ smallest bins at time $m'$.  Now consider the bins at time $m$.  For any $B_{i} \notin S_{\delta}$, let $x_{i}$ be the number of bins not in $S_{\delta}$ with the same load as $B_{i}$.  By the previous corollary, we know $$E(x_{i})=(n)(1-\delta)O(1/\sqrt{\log n})=O(n/\sqrt{\log n}).$$  So $$E\left(\sum_{B_{i} \notin S_{\delta}} x_{i}\right)= O(n^{2}/\sqrt{\log n}).$$   We can use Markov's Inequality to see that $$Pr\left( \sum_{B_{i} \notin S_{i}} x_{i}> a\right)=O(n^{2}/a\sqrt{\log n}).$$
 
 Now, suppose there were $y$ bins with the same load.  Then the sum would be at least $y(y-1)/2$ (as there would be at least $\binom{y}{2}$ pairs with the same load).  The probability of that happening is $$O(n^{2}/y^{2}\sqrt{\log n}).$$  If $y$ is such that $$\frac{n}{(\log n)^{1/4}}=o(y),$$ this probability goes to zero. 
\qquad\end{proof}
\section{Conclusion}

We now have a large body of knowledge about GREEDY's distribution.  We understand the behavior of small bins when $m$ is bounded in terms of $n$, subsets of small bins for arbitrary $m$, and all bins when $m$ is large.  In fact, our understanding is strongest when $m$ is much larger than $n$; we have shown that the bin loads' relative positions will stabilize and gaps between them will increase.  

It seems clear that the algorithm's behavior will become more predictable as the number of balls increases.  This is entirely the opposite of FAIR, which behaves more like UNIFORM as the number of balls increases. The author would be interested to see more research in the $m$ greater than $n$ case, perhaps combining Theorems \ref{equal} and \ref{gambler} to generate a new theorem similar to Theorem \ref{alpha1} that holds bin-by-bin rather than setwise.

As mentioned in the introduction, GREEDY-type algorithms naturally arise in several settings.  Another line of research of interest is in applying these theoretical results to specific instances; for example, modifying GREEDY to model consumer behavior.  The author has made some preliminary investigations in this direction, which are promising.

\section*{Acknowledgments}
Thanks to Peter Shor for academic and financial support during both the research and writing phases of this project.  Thanks to Sachin Lodha and the team at Tata Research Development and Design Centre for introducing the author to this problem and the first steps in its analysis.  Thanks to Joel Spencer for many helpful conversations.  Thanks to the National Science Foundation, who in part supported this work under Award No. 1004382

\bibliographystyle{siam.bst}
\bibliography{UnbalAl.bib}

\end{document}